\documentclass[12pt,a4paper]{article}
\usepackage[american]{babel}
\usepackage{amsmath,babel}
\usepackage{amssymb,latexsym}
\usepackage{amstext}
\usepackage{graphicx}
\usepackage{epsfig}
\usepackage[latin1]{inputenc}
\usepackage{amsfonts}

\newtheorem{theorem}{Theorem}

\newtheorem{definition}[theorem]{Definition}
\newtheorem{example}[theorem]{Example}

\newtheorem{lemma}[theorem]{Lemma}

\newenvironment{proof}[1][Proof]{\noindent\textbf{#1.} }{$\Box$\\}

\begin{document}
	
	\title{Necessary players and values}
	\author{J.C. Gon\c{c}alves-Dosantos$^1$, I. Garc\'{\i}a-Jurado$^1$, \\J. Costa$^2$, J.M. Alonso-Meijide$^3$}
	\date{{\today}}
	\maketitle
	
		\footnotetext[1]{Grupo MODES, CITIC and Departamento de Matem\'aticas, Universidade da Coru\~{n}a, Campus de Elvi\~{n}a, 15071 A Coru\~{n}a, Spain.} \footnotetext[2]{Grupo MODES, Departamento de Matem\'aticas, Universidade da Coru\~{n}a,
			Campus de Elvi\~{n}a, 15071 A Coru\~{n}a, Spain.}
	\footnotetext[3]{Grupo MODESTYA, Departamento de Estat\'{\i}stica, An\'alise Matem\'atica e
		Optimizaci\'on, Universidade de Santiago de Compostela, Campus de Lugo, 27002 Lugo, Spain.}

\begin{abstract}
In this paper we introduce the $\Gamma$ value, a new value for cooperative games with transferable utility. We also provide an axiomatic characterization of the $\Gamma$ value based on a property concerning the so-called necessary players. A necessary players of a game is one without which the characteristic function is zero. We illustrate the performance of the $\Gamma$ value in a particular cost allocation problem that arises when the owners of the apartments in a building plan to install an elevator and share its installation cost; in the resulting example we compare the proposals of the $\Gamma$ value, the equal division value and the Shapley value in two different scenarios. In addition, we propose an extension of the $\Gamma$ value for cooperative games with transferable utility and with a coalition structure. Finally, we provide axiomatic characterizations of the coalitional $\Gamma$ value and of the Owen and Banzhaf-Owen values using alternative properties concerning necessary players.
\end{abstract}
	

\section{Introduction}
The Shapley value, introduced in Shapley (1953), is a rule for distributing the benefits that a set of agents $N$ can generate, taking into account the contributions of each agent and of each possible subset of $N$. The Shapley value is one of the most important solutions of cooperative game theory and it has many applications in a wide variety of fields. For instance, in recent years the Shapley value has been applied to cancer research (see Albino et al., 2008), to machine learning (see Strumbelj and Kononenko, 2010), to data envelopment analysis (see Yang and Zhang, 2015), to image classification (see Gurran et al., 2016), to project management (see Bergantiños et al., 2018, and Gonçalves-Dosantos et al., 2020), etc. Moretti and Patrone (2008) is a survey explaining the transversality of the Shapley value.

The game theory literature provides many alternatives to the Shapley value, such as the nucleolus (Schmeidler, 1969), the Banzhaf Value (Owen, 1975), the $\tau$-value (Tijs, 1981), the equal-surplus division value (Driessen and Funaki, 1991) or, more recently, the consensus value (Ju et al., 2007) and the ie-Banzhaf value (Alonso-Meijide et al., 2019b). All those alternative values have appealing properties and could be used instead of the Shapley value. In order to decide what is the most appropriate value for a particular problem it is helpful to know the properties that are essentially connected to each value. This is why game theory is interested in the so-called characterizations: to characterize a value in a class of games is to find a set of properties so that it is the only value that fulfills them in that class. For instance, in Luchetti et al. (2010), two relevance indexes for genes are compared, one based on the Shapley value and the other based on the Banzhaf value, and for that purpose they are characterized in the corresponding class of games, the so-called microarray games.

In this article we introduce a value for cooperative games that results from proposing a new property for so-called necessary players that, in a way, corrects the properties for such players met by the Shapley and Banzhaf values. Informally, necessary players are those without whom the characteristic function of the game would be zero. These players have attracted the attention of game theorists for axiomatic studies in the last years. For instance, Alonso-Meijide et al. (2019a) and Béal and Navarro (2020) are two recent papers dealing with necessary players and characterizations. Apart from introducing a new value, in this paper we provide an axiomatic characterization of it, which allows to compare the new value with other solution concepts for cooperative games. Furthermore, we extend and characterize the new value for cooperative games with a coalition structure. A cooperative game with a coalition structure models those situations where the agents in a set $N$ aim to distribute the benefits they generate taking into account the contributions of each agent and of each possible subset of $N$, as well as a coalition structure (a partition of $N$) that conditions the distribution, in the sense that distribution among the classes of the partition is made first and, then, a distribution within those classes is performed. Cooperative games with a coalition structure have been applied in several fields like political analysis (see, for instance, Carreras and Puente, 2015), infrastructure management (see Costa, 2016), cost allocation (see Fragnelli and Iandolino, 2004), etc. The Owen value (Owen, 1977) and the Banzhaf-Owen value (Owen, 1982) are, respectively, the variations of the Shapley value and the Banzhaf value for cooperative games with a coalition structure. In this paper we also provide new characterizations of the Owen  and the Banzhaf-Owen values using properties involving necessary players.

The structure of this paper is as follows. In Section \ref{seccion2} we introduce the $\Gamma$ value, a new value for cooperative games. We also provide an axiomatic characterization of the $\Gamma$ value and illustrate its behaviour in a practical example that arises in a problem of sharing the costs of installing an elevator. In Section \ref{seccion3} we provide new characterizations of the Owen and Banzhaf-Owen values and introduce and characterize an extension of the $\Gamma$ value for cooperative games with a coalitional structure. We finish the paper with a section of concluding remarks.
	
%
	
	\section{Values and necessary players}
	\label{seccion2}
	
	A cooperative game is a pair $(N,v)$ given by a finite set of players $N$ and a characteristic function $v:2^N\rightarrow\mathbb{R}$, that assigns to each {coalition} $S\subseteq N$ a real number $v(S)$ that indicates the benefits that coalition $S$ is able to generate; by definition $v(\emptyset)=0$. We denote by $\mathcal{G}_N$ the family of all cooperative games with player set $N$.
	
	A \textit{value} for cooperative games is a map $f$ that assigns to every game $(N,v)\in\mathcal{G}_N$ a vector $f(N,v)\in\mathbb{R}^N$. Two of the most important values for cooperative games are the Shapley value (Shapley, 1953) and the Banzhaf value (Owen, 1975). A number of characterizations of these two values can be found on the literature. For example, Alonso-Meijide et al. (2019a) provides characterizations of those values using only three properties for each of them: two common properties and one extra property concerning the so-called {\em necessary players} that differs for the Shapley and the Banzhaf values. In this paper we concentrate on characterizations of values involving necessary players. Let us first remember the formal definition of a necessary player.
	
	\begin{definition}
	A player $i\in N$ is said to be necessary in the cooperative game $(N,v)$ if $v(S)=0$ for all $S\subseteq N\setminus\{ i\}$.
	\end{definition}
	
	
	In words, a necessary player is one without whom cooperation does not produce any results. In fact, notice that if $i$ is necessary in $(N,v)$, then $v(S)=\sum_{j\in S}v(\{ j\})=0$ for all $S\subseteq N\setminus\{ i\}$; hence, the game resulting after the elimination of $i$ is additive and null. Necessary players often arise in real situations. Take, for instance, the following example.
	
	\begin{example}
		Consider a council formed by three entities with $24,15$ and $9$ votes, respectively. Any proposal must receive at least $25$ votes to be approved. In the resulting voting game, it is easy to see that the entity with $24$ votes is a necessary player because, without it, the other two entities cannot get any proposals approved.
	\end{example}
	
	In some specific problems, as in the example above, the necessary players arise in a natural way and, therefore, a characterization based on such players can be relevant in deciding what value to use in those problems. We start by remembering the characterizations of the Shapley and Banzhaf values in Alonso-Meijide et al. (2019a) and some other preliminary material.
	The Shapley value $\varphi $ is defined as 
$$
\varphi _{i}\left( N,v\right) =\frac{1}{n}\sum_{S\subseteq N\setminus \{i\}}\frac{1}{{n-1 \choose s}}(v\left( S\cup \{i\}\right) -v\left( S\right) )
$$	
for all $(N,v)\in\mathcal{G}_N$ and all $i\in N$; $n$ and $s$ denote the cardinalities of $N$ and $S$, respectively.
The Banzhaf value $\beta $ is defined as
$$
\beta _{i}\left( N,v\right) =\frac{1}{2^{n-1}}\sum_{S\subseteq N\setminus \{i\}}(v\left( S\cup \{i\}\right) -v\left( S\right) )
$$
for all $(N,v)\in\mathcal{G}_N$ and all $i\in N$. Both the Shapley and the Banzhaf value are additive. This means that they satisfy the following condition.

\vspace*{0.25cm}

\noindent
\textbf{Additivity.} A value for cooperative games $f$ satisfies the property of additivity if for any pair of cooperative games $\left( N,v\right) ,\left( N,w\right) $ it holds that 
$$
f\left( N,v+w\right) =f\left( N,v\right) +f\left( N,w\right) .
$$ 
Additivity is a good property because, at the same time that it is natural and easily interpretable, it greatly facilitates the mathematical analysis of the values that comply with it and the calculation of such values; for instance, Benati et al. (2019) provides a method to approximate additive values in cooperative games that is useful when the number of players is large. 

Another reasonable property that is satisfied by the Shapley and Banzhaf value concerns null players. Remember that a null player of $(N,v)$ is an $i\in N$ such that $v(S)=v(S\cup\{ i\})$ for all $S\subseteq N\setminus\{ i\}$.

\vspace*{0.25cm}

\noindent	
\textbf{Null Player.} A value for cooperative games $f$ satisfies the property of null player if for any cooperative game $\left( N,v\right) $ and for any $i\in N$ null player of $\left( N,v\right) $, it holds that 
$f_{i}\left( N,v\right) =0$.

\vspace*{0.25cm}
	
Now let us see two alternative properties for necessary players introduced in Alonso-Meijide et al. (2019a) and the main result concerning them.

\vspace*{0.25cm}

\noindent
\textbf{Necessary Players Get the Weighted Mean.} A value for cooperative games $f$ satisfies the property of necessary players get the weighted mean if, for all cooperative game $\left( N,v\right) $ and for all $i\in N$ necessary player in $\left( N,v\right) $, it holds that 
$$
f_{i}\left( N,v\right) =\frac{1}{n}\sum_{S\subseteq N, i\in S}\frac{1}{{ n-1 \choose s-1}}v(S).
$$
\textbf{Necessary Players Get the Mean.}  A value for cooperative games $f$ satisfies the property of necessary players get the mean if, for all cooperative game $\left( N,v\right) $ and for all $i\in N$ necessary player in $\left( N,v\right) $, it holds that 
$$
f_{i}\left( N,v\right) =\frac{1}{2^{n-1}}\sum_{S\subseteq N, i\in S} v(S).
$$
Observe that the two properties above are similar. Both establish that a necessary player must receive the average of the values of the coalitions to which that player belongs, although the former takes into account the size of such coalitions and the latter does not.

\begin{theorem}
	\label{teoremauno}
(Alonso-Meijide et al., 2019a).\\
1. The Shapley value is the unique value for cooperative games that satisfies the properties of additivity, null player and necessary players get the weighted mean.\\
2. The Banzhaf value is the unique value for cooperative games that satisfies the properties of additivity, null player and necessary players get the mean.
\end{theorem}

Now we remember two widely known properties for values that will be relevant in the subsequent discussion.

\vspace*{0.25cm}

\noindent
\textbf{Efficiency.} A value for cooperative games $f$ satisfies the property of efficiency if for all cooperative game $\left( N,v\right)$, it holds that 
$$
\sum\limits_{i\in N}f_{i}\left( N,v\right) =v\left( N\right) .
$$

We say that players $i,j\in N$ are symmetric in $\left( N,v\right)\in{\cal G}_N$ if $v(S\cup \{i\})=v(S\cup \{j\})$ for every $S\subseteq N\setminus \left\{ i,j\right\} $.

\vspace*{0.25cm}

\noindent
\textbf{Symmetry.} A value for cooperative games $f$ satisfies the property of symmetry if for all cooperative game $\left( N,v\right) $ and for all $i,j\in N$ symmetric players in $\left( N,v\right) $, it holds that 
$$
f_{i}\left( N,v\right) =f_{j}\left( N,v\right) .
$$

It is well-known that the Shapley and Banzhaf values satisfy the symmetry property. However, only the Shapley value is efficient. In some problems, efficiency is not an essential property for a value, see for example microarray games in Lucchetti et al. (2010). In many cases, however, efficiency will be required for a value to make sense; this happens, for example, when we are faced with cost allocation problems. One question we can ask is whether there is a value that fulfills the necessary players get the mean property and the efficiency property. The answer is negative because those properties are incompatible. Indeed, assume that a value for cooperative games $f$ satisfies both properties and for every non-empty $S\subseteq N$ denote by $(N,e_S)$ the cooperative game in ${\cal G}_N$ given, for every $T\subseteq N$, by:
	\begin{equation}
	\label{basecanonica}
e_S(T)=\left\{
	\begin{array}{cc}
	1&\mbox{if $T=S$,}\\
	0&\mbox{otherwise.}\\
	\end{array}
		\right.
		\end{equation}
Since $f$ satisfies efficiency, it holds that
\begin{equation}
\label{inc1}
\sum\limits_{i\in N}f_{i}\left( N,e_N\right) =1.
\end{equation}
Notice now that every $i\in N$ is necessary in $( N,e_N)$ and then, since $f$ satisfies the necessary players get the mean property, it holds that
\begin{equation}
\label{inc2}
\sum\limits_{i\in N}f_{i}\left( N,e_N\right) =\sum\limits_{i\in N}\frac{1}{2^{n-1}}=\frac{n}{2^{n-1}}.
\end{equation}
Observe that (\ref{inc1}) and (\ref{inc2}) are incompatible for $n>2$, which implies that necessary players get the mean and efficiency are incompatible properties. Such incompatibility vanishes when we consider the next weak version of the former property.

\vspace*{0.25cm}

\noindent
\textbf{(Weak) Necessary Players Get the Mean.}  A value for cooperative games $f$ satisfies the (weak) necessary players get the mean property if, for all cooperative game $\left( N,v\right)$ with $v(N)=0$ and for all $i\in N$ necessary player in $\left( N,v\right) $, it holds that 
$$
f_{i}\left( N,v\right) =\frac{1}{2^{n-1}}\sum_{S\subseteq N, i\in S} v(S).
$$
%
With this new property we can prove the following proposition.

\begin{theorem}
\label{teoremados}
There exists a unique value for cooperative games that satisfies the properties of additivity, (weak) necessary players get the mean, efficiency and symmetry. This value that we denote by $G$ is given, for all $(N,v)\in{\cal G}_N$ and all $i\in N$, by:
\begin{equation}
\label{G}
G _{i}\left( N,v\right) =\frac{1}{2^{n-1}}\left(\sum_{S\subset N, i\in S} v(S)-\sum_{S\subset N, i\not\in S} \frac{s}{n-s}v(S)\right)+\frac{v(N)}{n}.
\end{equation}
\end{theorem}
\begin{proof}
	(Existence). It is clear that $G$ satisfies additivity. To check that it satisfies the (weak) necessary players get the mean property, take a cooperative game $\left( N,v\right)$ with $v(N)=0$ and such that $i\in N$ is a necessary player in $\left( N,v\right) $. Then expression (\ref{G}) reduces to
	$$
	G_i(N,v)=\frac{1}{2^{n-1}}\sum_{S\subset N, i\in S}v(S) =\frac{1}{2^{n-1}}\sum_{S\subseteq N, i\in S} v(S).\\
	$$
	To check that $G$ satisfies efficiency notice that, for every cooperative game $(N,v)$,
	\begin{eqnarray}
	\nonumber
	\sum_{i\in N}G_i(N,v)&=&\frac{1}{2^{n-1}}\sum_{i\in N}\left(\sum_{S\subset N, i\in S} v(S)-\sum_{S\subset N, i\not\in S} \frac{s}{n-s}v(S)\right)+v(N)\\
	\nonumber
	&=&\frac{1}{2^{n-1}}\left(\sum_{S\subset N}s v(S)-\sum_{S\subset N} (n-s)\frac{s}{n-s}v(S)\right)+v(N)\\
	\nonumber
	&=&v(N).
	\end{eqnarray}
	To check that $G$ satisfies symmetry take a cooperative game $(N,v)$ and a pair of symmetric players in $(N,v)$ $i,j\in N$.
	Notice that
	{\footnotesize
	\begin{eqnarray}
	\nonumber
 \sum_{S\subset N, i\in S} v(S)-\sum_{S\subset N, i\not\in S} \frac{s}{n-s}v(S)&=&\sum_{S\subseteq N\setminus \{ i,j\}}\left(  v(S\cup \{i\})\right)+\sum_{S\subset N\setminus \{ i,j\}}\left(  v(S\cup \{i, j\})\right)\\
\nonumber
	&&-\sum_{S\subseteq N\setminus\{ i,j\}}\left( \frac{s}{n-s}v(S)+ \frac{s+1}{n-s-1}v(S\cup \{j\})\right).
\end{eqnarray}
}
%
	\noindent
	Now, since $i,j$ are symmetric in $(N,v)$, the last expression is equal to
	
	{\scriptsize
		$$\sum_{S\subseteq N\setminus \{ i,j\}}\left(  v(S\cup \{j\})\right)+\sum_{S\subset N\setminus \{ i,j\}}\left(  v(S\cup \{i, j\})\right)-\sum_{S\subseteq N\setminus\{ i,j\}}\left( \frac{s}{n-s}v(S)+ \frac{s+1}{n-s-1}v(S\cup \{i\})\right)$$}
	
	\noindent
	and then it is clear that $G_i(N,v)=G_j(N,v)$.
	
	\noindent (Uniqueness). Take $f$, a value for cooperative games that satisfies efficiency, symmetry, (weak) necessary players get the mean and additivity and take a cooperative game $(N,v)$. We prove now that $f(N,v)=G (N,v)$. Indeed, consider the canonical basis of the vector space of characteristic functions of cooperative games with set of players $N$: $\{ e_S\}_{S\in 2^N\setminus\emptyset}$ (see expression (\ref{basecanonica})). Observe that $v$ can be written in a unique way as a linear combination of the elements of the canonical basis: $v=\sum_{S\in 2^N\setminus\emptyset} v(S) e_S$. Since $f$ satisfies additivity,
	$$f(N,v)=\sum_{S\in 2^N\setminus\emptyset} f(N,v(S)e_S).$$
	Note that efficiency, symmetry and (weak) necessary players get the mean characterize a unique value in the class $\{ (N,v(S)e_S)\ |\ S\subset N, S\neq\emptyset\}$. Besides, efficiency and symmetry characterize a unique value for $(N,v(N)e_N)$.  Hence $f(N,v)=G (N,v)$.
\end{proof}

Surprisingly enough, the new value $G$ introduced in Proposition \ref{teoremados} looks a lot like the e-Banzhaf value defined in Alonso-Meijide et al. (2019b) but it is not the same, because $\frac{n-s}{s}$ is changed by $\frac{s}{n-s}$ and, moreover, those two parameters do not multiply the same summands in the expressions of $G$ and of the e-Banzhaf value. Indeed, such parameters do not seem to have a clear interpretation from the point of view of fairness, which leads us to think that perhaps the (weak) necessary players get the mean property should be reformulated. 
In fact, it is more reasonable to ask that a necessary player be entitled to the average of the per capita values of the coalitions that contain it rather than the average of the values of those coalitions; in fact, such players are necessary for coalitions to have a value other than zero, but they require the other coalition members to generate such a value. Thus we propose the new property formulated below.

\vspace*{0.25cm}

\noindent
%
\textbf{Necessary Players Get the Per Capita Mean.}  A value for cooperative games $f$ satisfies the necessary players get the per capita mean property if, for all cooperative game $\left( N,v\right)$ with $v(N)=0$ and for all $i\in N$ necessary player in $(N,v)$, it holds that 
$$
f_{i}\left( N,v\right) =\frac{1}{2^{n-1}}\sum_{S\subseteq N, i\in S}\frac{v(S)}{s}.
$$

\vspace*{0.25cm}

The next result introduces and characterizes a new value for cooperative games.

\begin{theorem}
	\label{nuevo valor}
	There exists a unique value for cooperative games that satisfies the properties of additivity, necessary players get the per capita mean, efficiency and symmetry. This value that we denote by $\gamma$ is given, for all $(N,v)\in{\cal G}_N$ and all $i\in N$, by:
	\begin{equation}
	\label{Gamma}
	\gamma_{i}\left( N,v\right) =\frac{1}{2^{n-1}}\left(\sum_{S\subset N, i\in S} \frac{v(S)}{s}-\sum_{S\subset N, i\not\in S} \frac{v(S)}{n-s}\right)+\frac{v(N)}{n}.
	\end{equation}
\end{theorem}
\begin{proof}
		(Existence). It is clear that $\gamma$ satisfies additivity. To check that it satisfies the necessary players get the per capita mean property take a cooperative game $\left( N,v\right)$ with $v(N)=0$ and such that $i\in N$ is a necessary player in $\left( N,v\right) $. Then expression (\ref{Gamma}) reduces to
		$$
		\gamma_i(N,v)=\frac{1}{2^{n-1}}\sum_{S\subset N, i\in S}\frac{v(S)}{s}=\frac{1}{2^{n-1}}\sum_{S\subseteq N, i\in S}\frac{v(S)}{s}.\\
		$$
		To check that $\gamma$ satisfies efficiency notice that, for every cooperative game $(N,v)$,
		\begin{eqnarray}
		\nonumber
		\sum_{i\in N}\gamma_i(N,v)&=&\frac{1}{2^{n-1}}\sum_{i\in N}\left(\sum_{S\subset N, i\in S} \frac{v(S)}{s}-\sum_{S\subset N, i\not\in S} \frac{v(S)}{n-s}\right)+v(N)\\
		\nonumber
		&=&\frac{1}{2^{n-1}}\left(\sum_{S\subset N}s \frac{v(S)}{s}-\sum_{S\subset N} (n-s)\frac{v(S)}{n-s}\right)+v(N)\\
		\nonumber
		&=&v(N).
		\end{eqnarray}
		To check that $\gamma$ satisfies symmetry take a cooperative game $(N,v)$ and a pair of symmetric players in $(N,v)$ $i,j\in N$.
		Notice that
		
		$$\sum_{S\subset N, i\in S} \frac{v(S)}{s}-\sum_{S\subset N, i\not\in S} \frac{v(S)}{n-s}=$$
		{\footnotesize
			$$\sum_{S\subseteq N\setminus \{ i,j\}}  \frac{v(S\cup \{i\})}{s+1}+\sum_{S\subset N\setminus \{ i,j\}}\frac{v(S\cup \{i, j\})}{s+2}-\sum_{S\subseteq N\setminus\{ i,j\}}\left( \frac{v(S)}{n-s}+ \frac{v(S\cup \{j\})}{n-s-1}\right).$$}
		
		%
		\noindent
		Now, since $i,j$ are symmetric in $(N,v)$, the last expression is equal to
		
			{\footnotesize
			$$\sum_{S\subseteq N\setminus \{ i,j\}}  \frac{v(S\cup \{j\})}{s+1}+\sum_{S\subset N\setminus \{ i,j\}}\frac{v(S\cup \{i, j\})}{s+2}-\sum_{S\subseteq N\setminus\{ i,j\}}\left( \frac{v(S)}{n-s}+ \frac{v(S\cup \{i\})}{n-s-1}\right)$$}
		
		\noindent
		and then it is clear that $\gamma_i(N,v)=\gamma_j(N,v)$.
		
		\noindent (Uniqueness). Take $f$, a value for cooperative games that satisfies efficiency, symmetry, necessary players get the per capita mean and additivity and take a cooperative game $(N,v)$. We prove now that $f(N,v)=\gamma (N,v)$. Indeed, consider the basis of the vector space of characteristic functions of  cooperative games with set of players $N$: $\{ e_S\}_{S\in 2^N\setminus\emptyset}$ (see expression (\ref{basecanonica})). Observe that  $v$ can be written in a unique way as a linear combination of the elements of the basis: $v=\sum_{S\in 2^N\setminus\emptyset} v(S) e_S$. Since $f$ satisfies additivity,
		$$f(N,v)=\sum_{S\in 2^N\setminus\emptyset} f(N,v(S)e_S).$$
		Notice that efficiency, symmetry and necessary players get the per capita mean characterize a unique value in the class of games $\{ (N,v(S)e_S)\ |\ S\subset N, S\neq\emptyset\}$. Besides, efficiency and symmetry characterize a unique value for $(N,v(N)e_N)$.  Hence $f(N,v)=\gamma (N,v)$.
\end{proof}

A very desirable property for values for cooperative games is the invariance to S-equivalence, which we remember below. Two cooperative games with the same sets of players $(N,v)$ and $(N,w)$ are said to be $S$-equivalent if there exist $a\in\mathbb{R}$ with $a>0$ and $b\in\mathbb{R}^N$ such that, for every $T\subseteq N$, it holds that
$$w(T)=av(T)+\sum_{j\in T}b_j.$$
When $(N,v)$ and $(N,w)$ are $S$-equivalent we can transform $v$ into $w$ simply by changing the scale and translating the players' utilities. In these conditions it seems reasonable to ask a value for cooperative games $f$ that $f(N,v)$ is transformed into $f(N,w)$ by doing the corresponding change of scale and translations. 

\vspace*{0.25cm}

\noindent \textbf{Invariance to $S$-equivalence (INV)}$\mathbf{.}$ A value for  cooperative games $f$ satisfies invariance to S-equivalence if for any pair of S-equivalent  cooperative games $(N,v)$ and $(N,w)$ such that $w(T)=av(T)+\sum_{j\in T}b_j$ for all $T\subseteq N$ (with $a\in\mathbb{R}$, $a>0$ and $b\in\mathbb{R}^N$) it holds that, for every $i\in N$,
\begin{equation*}
f_i(N,w)=af_i(N,v)+b_i.
\end{equation*}

Unfortunately, the value $\gamma$ defined by (\ref{Gamma}) is not invariant to $S$-equivalence. Then, we make an adjustment of $\gamma$ that leads us to the $\Gamma$ value for cooperative games that we define below.

\begin{definition} 
The $\Gamma$ value for cooperative games is given for every $(N,v)\in{\cal G}_N$ and every $i\in N$ by:
\begin{equation}
\label{Gammacap}
\Gamma_i(N,v)=v(\{i\})+\gamma_i(N,v^0),
\end{equation}
where $v^0(S)=v(S)-\sum_{j\in S}v(\{j\})$ for all $S\subseteq N$.
\end{definition}

It is easy to check that $\Gamma$ satisfies the invariance to $S$-equivalence. In order to characterize it, we introduce below a new property concerning the necessary players.

\vspace*{0.25cm}

\noindent 
\textbf{Necessary Players Get the $0$-Normalized Per Capita Mean.}  A value for cooperative games $f$ satisfies the necessary players get the $0$-normalized per capita mean property if, for all cooperative game $\left( N,v\right)$ with $v(N)=\sum_{j\in N}v(\{j\})$ and for all $i\in N$ necessary player in $(N,v)$, it holds that 
$$
f_{i}\left( N,v\right) =v(\{i\})+\frac{1}{2^{n-1}}\sum_{S\subseteq N, i\in S}\frac{v^0(S)}{s}.
$$

\begin{theorem}
	\label{nuevo valor2}
	$\Gamma$ is the unique value for cooperative games that satisfies the properties of additivity, necessary players get the  $0$-normalized per capita mean, efficiency and symmetry. 
\end{theorem}
\begin{proof}
	(Existence). Since $\gamma$ satisfies additivity, efficiency and symmetry, it is clear that $\Gamma$ also satisfies those properties. To check that it fulfils the necessary players get the  $0$-normalized per capita mean property  take a cooperative game $\left( N,v\right)$ with $v(N)=\sum_{j\in N}v(\{j\})$ and such that $i\in N$ is a necessary player in $\left( N,v\right) $. Then expression (\ref{Gammacap}) reduces to
	$$
	\Gamma_i(N,v)=v(\{i\})+\frac{1}{2^{n-1}}\sum_{S\subset N, i\in S}\frac{v^0(S)}{s}=v(\{i\})+\frac{1}{2^{n-1}}\sum_{S\subseteq N, i\in S}\frac{v^0(S)}{s}.\\
	$$
	\noindent (Uniqueness). Take $f$ a value for cooperative games that satisfies efficiency, symmetry, necessary players get the  $0$-normalized per capita mean  and additivity and take a cooperative game $(N,v)$. We prove now that $f(N,v)=\Gamma (N,v)$. Indeed, consider the basis of the vector space of characteristic functions of cooperative games with set of players $N$ given by: 
	$$ \{ e_{\{i\}}+e_N\ |\ i\in N\}\cup \{e_S\ |\ S\in 2^N,|S|\geq 2\}.$$
	Observe that  $v$ can be written in a unique way as a linear combination of the elements of this basis. Since $f$ satisfies additivity and, moreover, the properties of efficiency, symmetry and necessary players get the  $0$-normalized per capita mean  characterize a unique value in the games of the basis, the proof is concluded.
\end{proof}

Now we analyse an example in order to make some comments on the $\Gamma$ value. It is based on a similar example in Alonso-Meijide et al. (2020).

\begin{example}
Consider a three-storey building with one apartment on each floor, the three apartments having the same surface. The three corresponding owners have agreed to install an elevator and share the corresponding cost. Such a cost is 120 (in
thousands of euros), 50 of which correspond to the machine, 40 to the works to make the hollow of the elevator (a fixed cost of 10 plus a cost of 10 for the owner of the apartment in the first floor that is incremented by 10 for the owner of the apartment in the second
floor and by an additional 10 for the owner of the apartment in the third floor), and 30 to the works to be done on each floor to allow access to the elevator (10 in each of them). According
to this, the cost $c(i)$ in which each player is involved is:

\begin{itemize}
	\item 50 (machine) + 10 (floor) + 20 (hollow) = 80 for $i=1$, the player of the
	first floor,
	
	\item 50 (machine) + 10 (floor) + 30 (hollow) = 90 for $i=2$, the player of the
	second floor,
	
	\item 50 (machine) + 10 (floor) + 40 (hollow) = 100 for $i=3$, the player of the
	third floor.
\end{itemize}
The rest of the corresponding cost game is given by: $c(\{ 1,2\})=100$, $c(\{ 1,3\})=c(\{ 2,3\})=110$, $c(N)=120$. Table \ref{table1} below shows the distribution of costs for each of the apartments
according to the Egalitarian value, the Shapley value and $\Gamma$. In European city centres it is common to find buildings coping with situations like the one described in this example. It is not uncommon for the owners of the lower floors to be less favourable to installing an elevator because of the costs involved. According to Spanish legislation, when owners decide to make an investment in the common elements of a building, the corresponding costs will be distributed in proportion to the owners' shares (which, in turn, sometimes depend only on the surface areas of the apartments). Therefore, the distribution due to the Egalitarian value will be the one proposed by the legislation in some occasions. Note that the proposed Shapley value and $\Gamma$ distributions tend to favour the owners of the lower floors. In short, $\Gamma$ seems to be the least controversial distribution in view of the usual dynamics of homeowners' communities, because it tends to favour the owners of the lowest floor, who are usually the most reluctant to bear the costs of installing an elevator.

\vspace*{0.5cm}

\begin{table}[htbp]
	\begin{center}
		\begin{tabular}{|c|c|c|c|}
			\hline
			& Egalitarian & Shapley &$\Gamma$\\ \hline
			$1$ & 40 & 33.3333 & 32.5 \\ \hline
			$2$ & 40 & 38.3333 & 38.75 \\ \hline
			$3$ & 40 & 48.3333 & 48.75 \\ \hline
		\end{tabular}%
	\end{center}
	\caption{The Egalitarian value, the Shapley value and $\Gamma$ for $(N,c)$}
	\label{table1}
\end{table}
It is not uncommon that in real situations such as those described in this example not all the owners are in favour of the elevator. When this occurs, sometimes the elevator will not be installed immediately even if the owners in favour of it have a majority. The reason for this is that the unfavourable owners (generally those on the lower floors) may refuse to pay the financial amounts due to them and the owners' community can only force them to do so by initiating legal proceedings which may be long, economically costly and which, moreover, may profoundly damage coexistence in the building. The practical consequence of this is that negotiations often take place within the owners' community to try to ensure that the installation of the elevator is possible without damaging coexistence in the building. One possible solution is that the owners not in favour of the elevator give up its service; this means that the elevator will not have stops on the corresponding floors, so that the works to give access to the elevator on those floors will not be necessary and the total cost of the installation will be lower. Assume, for instance that in the three-storey building in this example the owner of the apartment in the first floor is not in favour to install the elevator and, moreover, declares that he will not pay any costs unless a court decision obliges him to do so. Negotiation in the community may propose that the elevator does not serve the first floor. In that case, the cost $d(i)$ in which each player is involved is:

\begin{itemize}
	\item 0 for $i=1$, the player of the first floor,
	
	\item 50 (machine) + 10 (floor) + 30 (hollow) = 90 for $i=2$, the player of the
	second floor,
	
	\item 50 (machine) + 10 (floor) + 40 (hollow) = 100 for $i=3$, the player of the
	third floor.
\end{itemize}
The rest of the corresponding cost game is given by: $d(\{ 1,2\})=90$, $d(\{ 1,3\})=100$, $d(\{ 2,3\})=110$, $d(N)=110$. Table \ref{table2} below shows the distribution of costs for each of the apartments according to the Egalitarian value, the Shapley value and $\Gamma$. Note that the distribution given by the Egalitarian rule does not seem to facilitate the agreement on the installation of the elevator because the owner of the first floor will continue to pay a considerable amount and, in addition, will give up the service of the elevator. The distributions given by the Shapley value and by $\Gamma$, however, do seem to facilitate a final settlement. According to the Shapley value, the owner of the first floor will waive elevator service but pay nothing in return. According to $\Gamma$, the owner of the first floor will even receive a small compensation for the inconvenience caused to him by the works and the installation.

\vspace*{0.5cm}

\begin{table}[htbp]
	\begin{center}
		\begin{tabular}{|c|c|c|c|}
			\hline
			& Egalitarian & Shapley &$\Gamma$\\ \hline
			$1$ & 36.6666 & 0 & -6.6666 \\ \hline
			$2$ & 36.6666 & 50 & 53.3333 \\ \hline
			$3$ & 36.6666 & 60 & 63.3333 \\ \hline
		\end{tabular}%
	\end{center}
	\caption{The Egalitarian value, the Shapley value and $\Gamma$ for $(N,d)$}
	\label{table2}
\end{table}

\end{example}

\section{Coalitional values and necessary players}	
\label{seccion3}
	In this section we extend the $\Gamma$ value to cooperative games with a coalition structure. We start by remembering the mean features concerning that model.
	
	We  denote by $P(N)$ the set of all partitions of a finite set $N$. Each $P\in P(N)$, of the form $P=\{P_1,\dots,P_m\}$, is called a \textit{coalition structure} on $N$. We call unions of $P$ to its elements $P_1,\dots,P_m$. We denote by $M$ the set $\{1,...,m\}$.

	A \textit{cooperative game with a coalition structure} is a triple $(N,v,P)$ where $(N,v)\in\mathcal{G}_N$ and $P\in P(N)$.  $\mathcal{G}_N^{cs}$ denotes the family of all cooperative games with a coalition structure and with player set $N$. Note that the first two elements of a cooperative game with a coalition structure, $(N,v)$, characterize a cooperative game.
	
	By a \textit{coalitional value} we mean a map $g$ that assigns to every game with a coalition structure $(N,v,P)$ a vector $g(N,v,P)\in\mathbb{R}^N$ with components $g_i(N,v,P)$, $i\in N$. Two of the most important coalitional values are the Owen value (Owen, 1977) and the Banzhaf-Owen value (Owen, 1982). In a similar way to the Shapley and Banzhaf values, the value of a particular player is a weighted sum of his contributions. In the case of the Shapley and Banzhaf values all possible contributions are taken into account, but for the coalitional values only the contributions to some coalitions are used to compute the values.
	
	
The Owen value $\Phi$ is the  coalitional value defined by:
		\begin{equation*}
		 \Phi_i(N,v,P)=\frac{1}{m}\frac{1}{p_k}\sum_{R\subseteq M\backslash\{k\}} \sum_{T\subseteq
			P_k\backslash\{i\}}\frac{1}{{\binom{m-1}{r}}}\frac{1}{{\binom{p_k-1}{t}}} \Big[v(\bigcup\limits_{r\in R}P_r\cup T\cup\{i\})-v(\bigcup\limits_{r\in R}P_r\cup T)\Big]
		\end{equation*}
		for all $(N,v,P)\in\mathcal{G}_N^{cs}$ and all $i\in N$, where $P_k\in P$ is the union such that $i\in P_k$; $m$, $p_k$, $r$ and $t$ are the cardinalities of $M$, $P_k$, $R$ and $T$, respectively. 
		
		The \textit{Banzhaf--Owen value} $\Psi$ is the coalitional value defined as  
		\begin{equation*}
		\Psi_i(N,v,P)=\frac{1}{2^{m-1}}\frac{1}{2^{p_k-1}}\sum_{R\subseteq M\backslash\{k\}} \sum_{T\subseteq
			P_k\backslash\{i\}} \Big[v(\bigcup\limits_{r\in R}P_r\cup
		T\cup\{i\})-v(\bigcup\limits_{r\in R}P_r\cup T)\Big]
		\end{equation*}
		for all $(N,v,P)\in\mathcal{G}_N^{cs}$ and all $i\in N$, where $P_k\in P$ is the union such that $i\in P_k$; $m$, $p_k$, $r$ and $t$ are the cardinalities of $M$, $P_k$, $R$ and $T$, respectively.
%
%
%
%
	
	In the literature, we can find several characterizations of the Owen and the Banzhaf-Owen coalitional values; see for example Vázquez et al. (1997), Amer et al. (2002), Khmelnitskaya and Yanovskaya (2007), Alonso-Meijide et al. (2007), Casajus (2010) and Lorenzo-Freire (2016). We contribute to this research line providing a new characterization of these two coalitional values using necessary players. Only three properties are used in our results and the difference between them is the assigned payoff to necessary players.
	
%

	\bigskip 
	\noindent 
	\textbf{Necessary Players Get the Weighted Coalitional Mean}$\mathbf{.}$ A coalitional value $g$ satisfies the property of necessary players get the weighted coalitional mean if for any coalitional game $\left( N,v,P\right) $ and for any necessary player $i\in P_{k}$ in $(N,v)$, it holds that 

\begin{equation*}
g_{i}\left( N,v,P\right) =\frac{1}{m}\frac{1}{p_{k}}\sum_{R\subseteq
	M\backslash \{k\}}\sum_{T\subseteq P_{k}}\frac{1}{{\binom{m-1}{r}}}\frac{1}{{\binom{p_{k}-1}{t-1}}}v(\bigcup\limits_{r\in R}P_{r}\cup T).
\end{equation*}

	\bigskip
	
	\bigskip 
	\noindent 
	\textbf{Necessary Players Get the Coalitional Mean}$\mathbf{.}$ A coalitional value $g$ satisfies the property of necessary players get the coalitional mean if for any coalitional game $\left( N,v,P\right) $ and for any necessary player $i\in P_{k}$ in $(N,v)$, it holds that 
\begin{equation*}
g_{i}\left( N,v,P\right) =\frac{1}{2^{m-1}}\frac{1}{2^{p_{k}-1}}
\sum_{R\subseteq M\backslash \{k\}}\sum_{T\subseteq P_{k}}v(\bigcup\limits_{r\in
	R}P_{r}\cup T).
\end{equation*}
	
	\bigskip
	
	Both properties propose that a necessary player must receive the average worth over all coalitions that are compatible with the partitions (i.e., those that are formed by some complete unions and a subset of another union), but the first one takes into account the size of the coalitions while the second one assigns the same weight to all compatible coalitions. With these new properties we can prove the following results.
	
	\begin{theorem}
		The Banzhaf-Owen value is the unique coalitional value that satisfies the properties of additivity, null player and necessary players get the coalitional mean. \label{thbo}
	\end{theorem}
	
	\begin{proof}
		(Existence). It is known that the Banzhaf-Owen value satisfies additivity and null player. Now let us see that it satisfies the property of necessary players get the coalitional mean. Take a cooperative game with a coalition structure $(N,v,P)$ and take $i\in P_k$ a necessary player in $(N,v)$. Then the Banzhaf-Owen value is reduced to 
		\begin{eqnarray*}
			\Psi_i(N,v,P)&=&\frac{1}{2^{m-1}}\frac{1}{2^{p_k-1}}\sum_{R\subseteq M\backslash\{k\}} \sum_{T\subseteq
				P_k\backslash\{i\}} \Big[v(\bigcup\limits_{r\in R}P_r\cup
			T\cup\{i\})-v(\bigcup\limits_{r\in R}P_r\cup T)\Big] \\
			&=&\frac{1}{2^{m-1}}\frac{1}{2^{p_k-1}}\sum_{R\subseteq M\backslash\{k\}} \sum_{T\subseteq P_k\backslash\{i\}} v(\bigcup\limits_{r\in R}P_r\cup T\cup\{i\}) \\
			&=&\frac{1}{2^{m-1}}\frac{1}{2^{p_k-1}}\sum_{R\subseteq M\backslash\{k\}} \sum_{T\subseteq P_k} v(\bigcup\limits_{r\in R}P_r\cup T). \\
		\end{eqnarray*}
		
		(Uniqueness). For every $S\subseteq N$, $S\neq \emptyset$, the unanimity game $(N,u_S)$ is given, for every $T\subseteq N$, by: 
			\begin{equation}
		\label{unanimidad}
		u_S(T)=\left\{
		\begin{array}{cc}
		1&\mbox{if $S\subseteq T$,}\\
		0&\mbox{otherwise.}\\
		\end{array}
		\right.
		\end{equation}
		Take a coalitional value $g$ that satisfies additivity, null player and necessary players get the coalitional mean and take a cooperative game with a coalition structure $(N,v,P)$. We prove now that $g(N,v,P)=	\Psi(N,v,P)$. Given $S\subseteq N$, in the unanimity game $(N,u_S)$ every $i\in S$  is a necessary player and every $i\in N\backslash S$ is a null player. Let us fix $P$, a finite set $S\subseteq N$ and $c\in \mathbb{R}$.  By additivity it is sufficient to prove that for all $i\in N$, $g_i(N,c u_{S},P)=\Psi_i(N,c u_{S},P)$. If $i\in N\backslash S$, applying the null player property $g_i(N,c u_{S},P)=\Psi_i(N,c u_{S},P)=0$. If $i\in S$, applying neccesary players get the coalitional mean, we have that 
		$$g_i(N,c
		u_{S},P)=\Psi_i(N,c u_{S},P)=c \frac{1}{2^{m-1}}\frac{1}{2^{p_{k}-1}} \sum_{R\subseteq M\backslash k}\sum_{T\subseteq P_{k}}u_S(\bigcup\limits_{r\in R}P_{r}\cup  T),$$ where $P_k$ is the union such that $i\in P_k$.
	\end{proof}
	
	\begin{theorem}
		The Owen value is the unique coalitional value that satisfies the properties of additivity, null player and necessary players get the weighted coalitional mean. \label{tho}
	\end{theorem}
	
	\begin{proof}
		(Existence). It is known that the Owen value satisfies additivity and null player. Let us see that it satisfies the property of necessary players get the  weighted coalitional mean. Suppose that $i$ is a necessary player with $i\in P_k$; then the Owen value is  
		\begin{eqnarray*}
			\Phi_i(N,v,P)&=&\frac{1}{m}\frac{1}{p_k}\sum_{R\subseteq M\backslash\{k\}} \sum_{T\subseteq
				P_k\backslash\{i\}}\frac{1}{{\binom{m-1}{r}}}\frac{1%
			}{{\binom{p_k-1}{t}}} \Big[v(\bigcup\limits_{r\in R}P_r\cup T\cup\{i\})-v(\bigcup\limits_{r\in R}P_r\cup T)\Big] \\
			&=&\frac{1}{m}\frac{1}{p_k}\sum_{R\subseteq M\backslash\{k\}} \sum_{T\subseteq P_k\backslash\{i\}}%
			\frac{1}{{\binom{m-1}{r}}}\frac{1}{{\binom{p_k-1}{t}}%
			} v(\bigcup\limits_{r\in R}P_r\cup T\cup\{i\}) \\
			&=&\frac{1}{m}\frac{1}{p_k}\sum_{R\subseteq M\backslash\{k\}} \sum_{T\subseteq P_k}\frac{1}{{\binom{m-1}{r}}}\frac{1}{{\binom{p_k-1}{t-1}}} v(\bigcup\limits_{r\in R}P_r\cup T). \\
		\end{eqnarray*}
		
		(Uniqueness) Take a coalitional value $g$ that satisfies additivity, null player and necessary players get the weighted coalitional mean and take a cooperative game with a coalition structure $(N,v,P)$. We prove now that $g(N,v,P)=	\Phi(N,v,P)$. Given $S\subseteq N$, in the unanimity game $(N,u_{S})$ every $i\in S$ is a necessary player and every $i\in N\backslash S$ is a null player. Let us fix $P$, a finite set $S\subseteq N$ and $c\in \mathbb{R}$.  By additivity it is sufficient to prove that for all $i\in N$, $g_i(N,c u_{S},P)=\Phi_i(N,c u_{S},P)$. If $i\in N\backslash S$, applying the null player property 
		$$g_i(N,c u_{S},P)=	\Phi_i(N,c u_{S},P)=0.$$ If $i\in S$, applying necessary players get the weighted coalitional mean, we have that 
		$$g_i(N,c
		u_{S},P)=	\Phi_i(N,c u_{S},P)=c \frac{1}{2^{m-1}}\frac{1}{2^{p_{k}-1}} \sum_{R\subseteq
			M\backslash k}\sum_{T\subseteq P_{k}}u_S(\bigcup\limits_{r\in R}P_{r}\cup  T),$$ 
		where $P_k$ is the union such that $i\in P_k$.
	\end{proof}
	
	We are now willing to extend the $\Gamma$ value, defined in Section \ref{seccion2}, to cooperative games with a coalition structure.  We next remind some properties that are relevant for our aim.
	
	\vspace*{0.25cm}
	
	\noindent
	\textbf{Symmetry Inside Unions.} A coalitional value $g$ satisfies the property of symmetry inside unions if for all cooperative game with a coalition structure $\left( N,v, P\right)$, it holds that 
	$$
	g_{i}\left( N,v,P\right) =g_{j}\left( N,v,P\right) .
	$$
	for all $i,j$ symmetric players in $(N,v)$ with $i,j\in P_k$, $P_k\in P$.
	
	We say that unions $P_k, P_l\in P$ are symmetric in $\left( N,v,P\right)\in\mathcal{G}_N^{cs}$ if $v(S\cup P_k)=v(S\cup P_l)$, for every $S=\cup_{j\in R}P_j$ with $R\subseteq M\backslash\{k,l\}$.
	
	\vspace*{0.25cm}
	
	\noindent
	\textbf{Symmetry Among Unions.} A coalitional value $g$ satisfies the property of symmetry among unions if for all cooperative game with a coalition structure $\left( N,v, P\right)$, it holds that
	$$
	\sum_{i\in P_k}g_{i}\left( N,v,P\right) =\sum_{j\in P_r}g_{j}\left( N,v,P\right) 
	$$
	for all $P_k, P_r\in P$, symmetric unions in $(N,v,P)$.
%
%
	\vspace*{0.25cm}
	
	Given the properties of efficiency, additivity, symmetry inside unions and symmetry among unions one can expect to extend the $\Gamma$ value to cooperative games with a coalition structure and to characterize the new value using a property for necessary players that somewhat adapts the necessary players get the  $0$-normalized per capita mean property. First at all, let us see how to extend the $\gamma$ value, since the $\Gamma$ value depends on it.

	\begin{definition} 
		The $\gamma^C$ value for cooperative games with a coalition structure is given for every $(N,v,P)\in\mathcal{G}_N^{cs}$ and every $i\in P_k$ by:
		\begin{equation}
	\label{gamma coalicional}
	\begin{split}
	\gamma^C_{i}\left( N,v,P\right) &=\frac{1}{2^{m-1}}\frac{1}{2^{p_{k}-1}}\left(
	\sum_{R\subseteq M\backslash k}\sum_{T\subset P_{k}, i\in T}\frac{v(\bigcup\limits_{r\in
			R}P_{r}\cup T)}{t}\right.\\
		&\left.-\sum_{R\subseteq M\backslash k}\sum_{T\subset P_{k}, i\notin T, T\neq \emptyset}\frac{v(\bigcup\limits_{r\in
			R}P_{r}\cup T)}{p_k-t}\right)\\
	&+\frac{1}{2^{m-1}}\frac{1}{p_k}
	\left(\sum_{R\subset M, k\in R}\frac{v(\bigcup\limits_{r\in
			R}P_{r})}{r}-\sum_{R\subseteq M\backslash k}\frac{v(\bigcup\limits_{r\in
			R}P_{r})}{m-r}\right)+\frac{v(N)}{mp_k}.
	\end{split}
	\end{equation}
	\end{definition}
	
	Let us see that $\gamma^C$ is an reasonable extension of $\gamma$. To check it, we can see that $\gamma^C$ is a coalitional value of $\gamma$, that is $\gamma^C(N,v,P^n)=\gamma(N,v)$ for all $(N,v,P^n)\in\mathcal{G}_N^{cs}$ where $P^n=\{\{1\},...,\{n\}\}$. In fact
	
	\begin{equation*}
	\begin{split}
	\gamma^C_{i}\left( N,v,P^n\right) &=\frac{1}{2^{m-1}}
	\left(\sum_{R\subset M, k\in R}\frac{v(\bigcup\limits_{r\in
			R}P_{r})}{r}-\sum_{R\subseteq M\backslash k}\frac{v(\bigcup\limits_{r\in
			R}P_{r})}{m-r}\right)+\frac{v(N)}{mp_k}\\
	&=\frac{1}{2^{n-1}}
	\left(\sum_{S\subset N, i\in S}\frac{v(S)}{s}-\sum_{S\subseteq N\backslash \{i\}}\frac{v(S)}{n-s}\right)+\frac{v(N)}{n}=\gamma_{i}\left(N,v\right).
	\end{split}
	\end{equation*}
	
	The next lemma proves that $\gamma^C$ satisfies an interesting property for cooperative games with a coalition structure.
	
	\begin{lemma}
		The $\gamma^C$ value satisfies the quotient game property, i.e., that $$\sum_{i\in P_k}\gamma^C_{i}\left( N,v,P\right) =\gamma^C_{k}\left( M,v^P,P^{m}\right)$$ 
			for all $P_k\in P$, where $v^P(R)=v(\cup_{r\in R}P_r)$ for all $R\subseteq M$, and $P^m=\left\{ \left\{ 1\right\}  ,\ldots ,\left\{ m\right\} \right\}$.
	\end{lemma}
	\begin{proof}	
	 Take a cooperative game with a coalition structure $(N,v,P)$ and $P_k\in P$. Then
	
	\begin{equation*}
	\begin{split}
	&\sum_{i\in P_k}\gamma^C_{i}\left( N,v,P\right)\\
	 &=\frac{1}{2^{m-1}}\frac{1}{2^{p_{k}-1}}\sum_{i\in P_k}\left(
	\sum_{R\subseteq M\backslash k}\sum_{T\subset P_{k}, i\in T}\frac{v(\bigcup\limits_{r\in
			R}P_{r}\cup T)}{t}-\sum_{R\subseteq M\backslash k}\sum_{T\subset P_{k}, i\notin T, T\neq \emptyset}\frac{v(\bigcup\limits_{r\in
			R}P_{r}\cup T)}{p_k-t}\right)\\
	&+\frac{1}{2^{m-1}}\frac{1}{p_k}\sum_{i\in P_k}
	\left(\sum_{R\subset M, k\in R}\frac{v(\bigcup\limits_{r\in
			R}P_{r})}{r}-\sum_{R\subseteq M\backslash k}\frac{v(\bigcup\limits_{r\in
			R}P_{r})}{m-r}\right)+\sum_{i\in P_k}\frac{v(N)}{mp_k}\\
	&=\frac{1}{2^{m-1}}\frac{1}{2^{p_{k}-1}}\left(
	\sum_{R\subseteq M\backslash k}\sum_{T\subset P_{k}}\frac{tv(\bigcup\limits_{r\in
			R}P_{r}\cup T)}{t}-\sum_{R\subseteq M\backslash k}\sum_{T\subset P_{k}}\frac{(p_k-t)v(\bigcup\limits_{r\in
			R}P_{r}\cup T)}{p_k-t}\right)\\
	&+\frac{1}{2^{m-1}}
	\left(\sum_{R\subset M, k\in R}\frac{v(\bigcup\limits_{r\in
			R}P_{r})}{r}-\sum_{R\subseteq M\backslash k}\frac{v(\bigcup\limits_{r\in
			R}P_{r})}{m-r}\right)+\frac{v(N)}{m}\\
	&=\frac{1}{2^{m-1}}
	\left(\sum_{R\subset M, k\in R}\frac{v(\bigcup\limits_{r\in
			R}P_{r})}{r}-\sum_{R\subseteq M\backslash k}\frac{v(\bigcup\limits_{r\in
			R}P_{r})}{m-r}\right)+\frac{v(N)}{m}=\gamma^C_{k}\left( M,v^P,P^m\right).
	\end{split}
	\end{equation*}
		\end{proof}
	
	The quotient game is an interesting property because it guarantees that the total worth obtained by the players of a union coincides with the worth obtained by the union in the game played by the unions with the trivial coalition structure. Note that the Banzhaf-Owen value does not satisfy this property; however, Alonso-Meijide and Fiestras-Janeiro (2002) introduces the so-called symmetric coalitional Banzhaf value, which is an extension of the Banzhaf value to cooperative games with a coalition structure that satisfies the quotient game property. 
	
	In order to characterize $\gamma^C$, we introduce a new property for necessary players.

	\bigskip 
	\noindent
	\textbf{Necessary Players Get the Per Capita Coalitional Mean}$\mathbf{.}$ A coalitional value $g$ satisfies the property of necessary players get the per capita coalitional mean if for any coalitional game $\left( N,v,P\right) $ with $v(N)=0$ and for any necessary player $i\in P_{k}$ in $(N,v)$, it holds that 
	\begin{equation*}
	g_{i}\left( N,v,P\right) =\frac{1}{2^{m-1}}\left[\frac{1}{2^{p_{k}-1}}
	\sum_{R\subseteq M\backslash k}\sum_{T\subset P_{k}, i\in T}\frac{v(\bigcup\limits_{r\in
			R}P_{r}\cup T)}{t}+\frac{1}{p_k}
	\sum_{R\subseteq M, k\in R}\frac{v(\bigcup\limits_{r\in
			R}P_{r})}{r}\right]
	\end{equation*}
	where $t=|T|$ and $r=|R|$ for all $T\subset P_k$ and $R\subseteq M$.
	
	\vspace*{0.25cm}
	
\begin{theorem}
		The $\gamma^C$ value is the unique value for cooperative games with a coalition structure that satisfies the properties of additivity, necessary players get the per capita coalitional mean, efficiency, symmetry inside unions and symmetry among unions. 
		
\end{theorem}
\begin{proof}
	(Existence). It is clear that $\gamma^C$ satisfies additivity. To check that it satisfies the necessary players get the per capita coalitional mean property take a cooperative game with a coalition structure $\left( N,v,P\right)$ with $v(N)=0$ and such that $i\in P_k\subseteq N$ is a necessary player in $\left( N,v\right) $. Then expression (\ref{gamma coalicional}) reduces to
			\begin{eqnarray}
				\begin{split}
			\nonumber
			\displaystyle\gamma^C_i(N,v,P)&=\displaystyle\frac{1}{2^{m-1}}\frac{1}{2^{p_{k}-1}}\left(\sum_{R\subseteq M\backslash k}\sum_{T\subset P_{k}, i\in T}\frac{v(\bigcup\limits_{r\in
					R}P_{r}\cup T)}{t}\right)+\frac{1}{2^{m-1}}\frac{1}{p_k}
				\left(\sum_{R\subset M, k\in R}\frac{v(\bigcup\limits_{r\in
				R}P_{r})}{r}\right)\\
				\nonumber
			&\displaystyle= \frac{1}{2^{m-1}}\frac{1}{2^{p_{k}-1}}
			\sum_{R\subseteq M\backslash k}\sum_{T\subset P_{k}}\frac{v(\bigcup\limits_{r\in
					R}P_{r}\cup T)}{t}+\frac{1}{2^{m-1}}\frac{1}{p_k}
			\sum_{R\subseteq M}\frac{v(\bigcup\limits_{r\in
					R}P_{r})}{r}.
				\end{split}
			\end{eqnarray}
		
			To check that $\gamma^C$ satisfies symmetry inside coalitions take a cooperative game with a coalition structure $(N,v, P)$ and a pair of symmetric players in $(N,v)$ $i,j\in P_k$ with $P_k\in P$.
			Notice that, for a fixed $R\subseteq M\backslash k$,
			\begin{eqnarray}
			\begin{split}
			\nonumber
			&\sum_{T\subset P_{k}, i\in T}\frac{v(\bigcup\limits_{r\in R}P_{r}\cup T)}{t}-\sum_{T\subset P_{k}, i\notin T, T\neq \emptyset}\frac{v(\bigcup\limits_{r\in R}P_{r}\cup T)}{p_k-t}\\
				&=\sum_{T\subseteq P_{k}\backslash\{i,j\}}\frac{v(\bigcup\limits_{r\in R}P_{r}\cup T\cup\{i\})}{t+1}+\sum_{T\subset P_{k}\backslash\{i,j\}}\frac{v(\bigcup\limits_{r\in R}P_{r}\cup T\cup\{i,j\})}{t+2}\\
				&-\sum_{T\subseteq P_{k}\backslash\{i,j\}, T\neq \emptyset} \frac{v(\bigcup\limits_{r\in R}P_{r}\cup T)}{p_k-t}-\sum_{T\subseteq P_{k}\backslash\{i,j\}} \frac{v(\bigcup\limits_{r\in R}P_{r}\cup T\cup\{j\})}{p_k-t-1}.
				\end{split}
			\end{eqnarray}
			%
			\noindent
			Now, since $i,j$ are symmetric in $(N,v)$, the last expression is equal to
			
				\begin{eqnarray}
			\begin{split}
			\nonumber
			&	\sum_{T\subseteq P_{k}\backslash\{i,j\}}\frac{v(\bigcup\limits_{r\in R}P_{r}\cup T\cup\{j\})}{t+1}+\sum_{T\subset P_{k}\backslash\{i,j\}}\frac{v(\bigcup\limits_{r\in R}P_{r}\cup T\cup\{i,j\})}{t+2}\\
			&	-\sum_{T\subseteq P_{k}\backslash\{i,j\}, T\neq \emptyset} \frac{v(\bigcup\limits_{r\in R}P_{r}\cup T)}{p_k-t}-\sum_{T\subseteq P_{k}\backslash\{i,j\}} \frac{v(\bigcup\limits_{r\in R}P_{r}\cup T\cup\{i\})}{p_k-t-1}
				\end{split}
				\end{eqnarray}
			
			\noindent
			and then it is clear that $\gamma^C_i(N,v,P)=\gamma^C_j(N,v,P)$.
			
			Since $\gamma^C$ satisfies the quotient game property and it is a coalitional value of $\gamma$, then
			$$\sum_{i\in P_k}\gamma^C_{i}\left( N,v,P\right)=\gamma^C_{k}\left( M,v^P,P^m\right)=\gamma_k(M,v^P).$$
			Now, the efficiency and the symmetry properties of $\gamma$ imply that $\gamma^C$ satisfies symmetry among unions and efficiency.

			\noindent (Uniqueness). Take $g$, a value for cooperative games with a coalition structure that satisfies efficiency, symmetry inside unions, symmetry among unions, necessary players get the per capita coalitional mean and additivity, and take a cooperative game with a coalition structure $(N,v,P)$. We prove now that $g(N,v,P)=\gamma^C (N,v,P)$. Indeed, consider the basis of the vector space of characteristic functions of  cooperative games with set of players $N$ given by: $\{ e_S\}_{S\in 2^N\setminus\emptyset}$ (see expression (\ref{basecanonica})). Observe that  $v$ can be written in a unique way as a linear combination of the elements of the basis: $v=\sum_{S\in 2^N\setminus\emptyset} v(S) e_S$. Since $g$ satisfies additivity,
			$$g(N,v,P)=\sum_{S\in 2^N\setminus\emptyset} g(N,v(S)e_S,P).$$
			Notice that efficiency, symmetry inside unions, symmetry among unions, and necessary players get the per capita coalitional mean characterize a unique value in the class of games $\{ (N,v(S)e_S,P)\ |\ S\subset N, S\neq\emptyset\}$. Besides, efficiency, symmetry inside unions and symmetry among unions, characterize a unique value for $(N,v(N)e_N,P)$.  Hence $g(N,v,P)=\gamma^C (N,v,P)$.
\end{proof}

Now, in an analogous way as we obtain $\Gamma$ from $\gamma$, we introduce the following value.		
		
		\begin{definition} 
			The $\Gamma^C$ value for cooperative games with a coalition structure is given for every $(N,v,P)\in{\cal G}_N^{cs}$ and every $i\in P_k$ by:
			\begin{equation}
			\label{Gammaunions}
			\Gamma^C_i(N,v,P)=v(\{i\})+\frac{v(P_k)-\sum_{j\in P_k}v(\{j\})}{p_k}+\gamma^C_i(N,v^{0'},P),
			\end{equation}
			where $v^{0'}(S)=v(S)-\sum_{r\in R}v(P_r)-\sum_{j\in S\backslash (\cup_{r\in R}P_r)}v(\{j\})$ and $R=\{r\in M\ |\ P_r\subseteq S\}$ for all $S\subseteq N$.
		\end{definition}

	As for $\gamma^C$, we check that $\Gamma^C$ is a coalitional value of $\Gamma$. Take the cooperative game with a coalition structure $(N,v,P^n)$. Then
	
	\begin{equation*}
	\begin{split}
	\Gamma^C_{i}\left( N,v,P^n\right) &=v(\{i\})+\frac{v(P_k)-\sum_{j\in P_k}v(\{j\})}{p_k}\\
	&+\frac{1}{2^{m-1}}
	\left(\sum_{R\subset M, k\in R}\frac{v^{0'}(\bigcup\limits_{r\in
			R}P_{r})}{r}-\sum_{R\subseteq M\backslash k}\frac{v^{0'}(\bigcup\limits_{r\in
			R}P_{r})}{m-r}\right)+\frac{v^{0'}(N)}{mp_k}\\
	&=v(\{i\})+\frac{1}{2^{n-1}}
	\left(\sum_{S\subset N, i\in S}\frac{v^0(S)}{s}-\sum_{S\subseteq N\backslash \{i\}}\frac{v^0(S)}{n-s}\right)+\frac{v^0(N)}{n}\\
	&=\Gamma_{i}\left(N,v\right).
	\end{split}
	\end{equation*}
		
	Now we provide an axiomatic characterization of $\Gamma^C$. We start with a lemma concerning the quotient game property.
	
	\begin{lemma}
		The $\Gamma^C$ value satisfies the quotient game property, i.e., that $$\sum_{i\in P_k}\Gamma^C_{i}\left( N,v,P\right) =\Gamma^C_{k}\left( M,v^P,P^{m}\right)$$ 
		for all $P_k\in P$, where $v^P(R)=v(\cup_{r\in R}P_r)$ for all $R\subseteq M$, and $P^m=\left\{ \left\{ 1\right\}  ,\ldots ,\left\{ m\right\} \right\}$.
	\end{lemma}
	\begin{proof}
	 Take a cooperative game with a coalition structure $(N,v,P)$ and $i\in N$ such that $P_k\in P$. Then 
	
	\begin{equation*}
	\begin{split}
	 &\sum_{i\in P_k}\Gamma^C_{i}\left( N,v,P\right)\\
	 &=v(P_k)+\frac{1}{2^{m-1}}\frac{1}{2^{p_{k}-1}}\sum_{i\in P_k}\left(
	\sum_{R\subseteq M\backslash k}\sum_{T\subset P_{k}, i\in T}\frac{v^{0'}(\bigcup\limits_{r\in
			R}P_{r}\cup T)}{t}\right.\\
	&\left.-\sum_{R\subseteq M\backslash k}\sum_{T\subset P_{k}, i\notin T, T\neq \emptyset}\frac{v^{0'}(\bigcup\limits_{r\in
			R}P_{r}\cup T)}{p_k-t}\right)\\
		&+\frac{1}{2^{m-1}}\frac{1}{p_k}\sum_{i\in P_k}
	\left(\sum_{R\subset M, k\in R}\frac{v^{0'}(\bigcup\limits_{r\in
			R}P_{r})}{r}-\sum_{R\subseteq M\backslash k}\frac{v^{0'}(\bigcup\limits_{r\in
			R}P_{r})}{m-r}\right)+\sum_{i\in P_k}\frac{v^{0'}(N)}{mp_k}\\
		&=v(P_k)\\
	&+\frac{1}{2^{m-1}}\frac{1}{2^{p_{k}-1}}\left(
	\sum_{R\subseteq M\backslash k}\sum_{T\subset P_{k}}\frac{tv^{0'}(\bigcup\limits_{r\in
			R}P_{r}\cup T)}{t}-\sum_{R\subseteq M\backslash k}\sum_{T\subset P_{k}}\frac{(p_k-t)v^{0'}(\bigcup\limits_{r\in
			R}P_{r}\cup T)}{p_k-t}\right)\\
	&+\frac{1}{2^{m-1}}
	\left(\sum_{R\subset M, k\in R}\frac{v^{0'}(\bigcup\limits_{r\in
			R}P_{r})}{r}-\sum_{R\subseteq M\backslash k}\frac{v^{0'}(\bigcup\limits_{r\in
			R}P_{r})}{m-r}\right)+\frac{v^{0'}(N)}{m}\\
		&=v(P_k)+\frac{1}{2^{m-1}}
	\left(\sum_{R\subset M, k\in R}\frac{v^{0'}(\bigcup\limits_{r\in
			R}P_{r})}{r}-\sum_{R\subseteq M\backslash k}\frac{v^{0'}(\bigcup\limits_{r\in
			R}P_{r})}{m-r}\right)+\frac{v^{0'}(N)}{m}\\
		&=\Gamma^C_{k}\left( M,v^P,P^m\right).
	\end{split}
	\end{equation*}
		\end{proof}
	
		In order to characterize $\Gamma^C$, we introduce a new property for necessary players.
		
		\bigskip 
	\noindent
	\textbf{Necessary Players Get the $0$-Normalized Per Capita Coalitional Mean}$\mathbf{.}$ A coalitional value $g$ satisfies the property of necessary players get the $0$-normalized per capita coalitional mean  if for any coalitional game $\left( N,v,P\right) $ with $v(N)=\sum_{r\in M}v(P_r)$ and for any necessary player $i\in P_{k}$ in $(N,v)$, it holds that 
	\begin{equation*}
	\begin{split}
	g_{i}\left( N,v,P\right) &=v(\{i\})+\frac{v(P_k)-\sum_{j\in P_k}v(\{j\})}{p_k}\\
	&+\frac{1}{2^{m-1}}\frac{1}{2^{p_{k}-1}}\sum_{R\subseteq M\backslash k}\sum_{T\subset P_{k}}\frac{v^{0'}(\bigcup\limits_{r\in
			R}P_{r}\cup T)}{t}+\frac{1}{2^{m-1}}\frac{1}{p_k}
	\sum_{R\subseteq M}\frac{v^{0'}(\bigcup\limits_{r\in
			R}P_{r})}{r},
	\end{split}
	\end{equation*}
	where $t=|T|$ and $r=|R|$ for all $T\subset P_k$ and $R\subseteq M$.
	
	\vspace*{0.25cm}
	
	\begin{theorem}
		\label{Gammacoal}
		$\Gamma^C$ is the unique coalitional value for cooperative games with a coalition structure that satisfies the properties of additivity, necessary players get the $0$-normalized per capita coalitional mean, efficiency, symmetry inside unions and symmetry among unions.
	\end{theorem}
	\begin{proof}
		(Existence). Since $\gamma^C$ satisfies additivity and symmetry inside unions, it is clear that $\Gamma^C$ also satisfies those properties. To check that it fulfils the necessary players get the $0$-normalized per capita coalitional mean take a cooperative game with a coalition structure $\left( N,v, P\right)$ with $v(N)=\sum_{r\in M}v(P_r)$ and such that $i\in N$, with $i\in P_k\in P$, is a necessary player in $\left( N,v\right) $. Then expression (\ref{Gammaunions}) reduces to
		$$
		\Gamma^C_i(N,v,P)=v(\{i\})+\frac{v(P_k)-\sum_{j\in P_k}v(\{j\})}{p_k}\\
			$$
			$$
			+\frac{1}{2^{m-1}}\frac{1}{2^{p_{k}-1}}\left(
			\sum_{R\subseteq M\backslash k}\sum_{T\subset P_{k}}\frac{v^{0'}(\bigcup\limits_{r\in
			R}P_{r}\cup T)}{t}\right)+\frac{1}{2^{m-1}}\frac{1}{p_k}
		\left(\sum_{R\subset M}\frac{v^{0'}(\bigcup\limits_{r\in
				R}P_{r})}{r}\right).
		$$
		
		Since this solution satisfies the quotient game property and it is a coalitional value of $\Gamma$, for a cooperative game with a coalition structure $(N,v,P)$ and $P_k\in P$ it holds that
		$$\sum_{i\in P_k}\Gamma^C_{i}\left( N,v,P\right)=\Gamma^C_{k}\left( M,v^P,P^m\right)=\Gamma_k(M,v^P).$$
		Then it is easy to check that $\Gamma^C$ satisfies symmetry among unions and efficiency taking into account that $\Gamma$ satisfies efficiency and symmetry.

		\noindent (Uniqueness). Take $g$, a value for cooperative games with a coalition structure that satisfies efficiency, symmetry, necessary players get the $0$-normalized per capita coalitional mean and additivity and take a cooperative game with a coalition structure $(N,v,P)$. We prove now that $g(N,v,P)=\Gamma^C (N,v,P)$. Indeed, consider the basis of the vector space of characteristic functions of cooperative games with set of players $N$ given by: 
		$$ \{ e_{P_r}+e_N\ |\ r\in M\}\cup \{e_S\ |\ S\in 2^N,S\neq P_r, r\in M\}.$$
		Observe that  $v$ can be written in a unique way as a linear combination of the elements of this basis. Since $g$ satisfies additivity and, moreover, the properties of efficiency, symmetry and necessary players get the $0$-normalized per capita coalitional mean characterize a unique value in the games of the basis, the proof is concluded.
	\end{proof}

	\section{Concluding Remarks}
	
	Notice that $G$, $\gamma$ and $\Gamma$, the three values introduced in Section \ref{seccion2}, satisfy the properties of additivity, efficiency and symmetry and then they can be written using the formula provided in Ruiz et al. (1998). Moreover, it is clear that $G$ and $\gamma$ satisfy the property of coalitional monotonicity dealt with in Wang et al. (2019) and thus, in view of Theorem 3.2 in Wang et al. (2019), they belong to the family of ideal values. Moreover, it is not difficult to prove that $\Gamma$ also satisfies the property of coalitional monotonicity and then it is also an ideal value. We provide next such a proof; to start with, we remember the property of coalitional monotonicity.
	
	\vspace*{0.15cm}

\noindent  
\textbf{Coalitional Monotonicity.} A value for cooperative games $f$ satisfies the property
of coalitional monotonicity if for any pair of cooperative games $\left( N,v\right)$ and $\left( N,w\right)$ fulfilling that there exists $T\subseteq N$ with $v(T)>w(T)$ and $v(S)=w(S)$ for all $S\subseteq N$, $S\neq T$, it holds that 
	\begin{equation*}
	f_i(N,v)\geq f_i(N,w)
	\end{equation*}
	for all $i\in T$.
	
	\vspace*{0.15cm}
	
	\noindent 
	In view of expressions (\ref{G}) and (\ref{Gamma}), it is clear that $G$ and $\gamma$ satisfy the property of coalitional monotonicity. With respect to $\Gamma$ notice that, in view of expressions (\ref{Gamma}) and (\ref{Gammacap}), for every TU-game $(N,u)$ and for every $i\in N$, $\Gamma_i(N,u)$ can be written as:
	
	\begin{eqnarray}
	\nonumber
	\Gamma_i(N,u)&=&\frac{1}{2^{n-1}}\sum_{S\subset N, i\in S}\frac{1}{s} \left(u(S)-\sum_{j\in S}u(j)\right)\\
	\label{e3}
	&&-\frac{1}{2^{n-1}}\sum_{S\subset N, i\not\in S} \frac{1}{n-s}\left(u(S)-\sum_{j\in S}u(j)\right)\\
	\nonumber
	&&+\frac{1}{n}\left(u(N)-\sum_{j\in N}u(j)\right)+u(i).
	\end{eqnarray}
	Take now $\left( N,v\right)$, $\left( N,w\right)$ and $T\subseteq N$ as in the statement of coalitional monotonicity. Using (\ref{e3}) it is clear that if $T$ has two or more elements, then $\Gamma_i(N,v)\geq \Gamma_i(N,w)$ for all $i\in T$. Assume now that $T=\{ i\}$ ($i\in N$). According to (\ref{Gammacap}), the coefficients of $v(i)$ and $w(i)$ in $\Gamma_i(N,v)$ and $\Gamma_i(N,w)$ are identical and given by:
	\begin{eqnarray*}
		&&-\frac{1}{2^{n-1}}\sum_{S\subset N, i\in S, S\neq i}\frac{1}{s}-\frac{1}{n}+1\\
		&=&-\frac{1}{2^{n-1}}\sum_{s=2}^{n-1}\frac{1}{s}\left(\begin{array}{c}n-1\\s-1\end{array}\right)-\frac{1}{n}+1\\
		&=&-\frac{1}{2^{n-1}}\sum_{s=2}^{n-1}\frac{1}{n}\left(\begin{array}{c}n\\s\end{array}\right)-\frac{1}{n}+1\\
		&=&-\frac{1}{2^{n-1}}\frac{1}{n}(2^{n}-1-n-1)-\frac{1}{n}+1=\frac{n-3}{n}+\frac{2+n}{2^{n-1}n}.
	\end{eqnarray*}
	It is easy to check that $\frac{n-3}{n}+\frac{2+n}{2^{n-1}n}> 0$ for all $n\geq 1$, which implies that $\Gamma_i(N,v)> \Gamma_i(N,w)$ and completes the proof.
	
	We finish this paper with a remark on the relation between our new values and the equal division and the equal surplus division values, that we denote by $ED$ and $ESD$ (see, for instance Alonso-Meijide et al., 2020). It is clear that, for every $(N,v)\in{\cal G}_N$ and every $i\in N$,
	
	\begin{equation*}
	\gamma_{i}\left( N,v\right) =ED_i(N,v)+\frac{1}{2^{n-1}}\left(\sum_{S\subset N, i\in S} \frac{v(S)}{s}-\sum_{S\subset N, i\not\in S} \frac{v(S)}{n-s}\right),
	\end{equation*}
	
	\begin{equation*}
	\Gamma_{i}\left( N,v\right) =ESD_i(N,v)+\frac{1}{2^{n-1}}\left(\sum_{S\subset N, i\in S} \frac{v^0(S)}{s}-\sum_{S\subset N, i\not\in S} \frac{v^0(S)}{n-s}\right).
	\end{equation*}
	
	Now, if $ED^U$ and $ESD2^U$ are the extensions of $ED$ and $ESD$ for cooperative games with a coalition structure introduced in Alonso-Meijide et al. (2020), then for every $(N,v,P)\in {\cal G}_N^{cs}$ and every $i\in N$ it holds that
	
		\begin{equation*}
	\begin{split}
	&\gamma^C_{i}\left( N,v,P\right) =ED^U_i(N,v,P)\\
	&+\frac{1}{2^{m-1}}\frac{1}{2^{p_{k}-1}}\left(
	\sum_{R\subseteq M\backslash k}\sum_{T\subset P_{k}, i\in T}\frac{v(\bigcup\limits_{r\in
			R}P_{r}\cup T)}{t}-\sum_{R\subseteq M\backslash k}\sum_{T\subset P_{k}, i\notin T, T\neq \emptyset}\frac{v(\bigcup\limits_{r\in
			R}P_{r}\cup T)}{p_k-t}\right)\\
	&+\frac{1}{2^{m-1}}\frac{1}{p_k}
	\left(\sum_{R\subset M, k\in R}\frac{v(\bigcup\limits_{r\in
			R}P_{r})}{r}-\sum_{R\subseteq M\backslash k}\frac{v(\bigcup\limits_{r\in
			R}P_{r})}{m-r}\right),
	\end{split}
	\end{equation*}
	
	\begin{equation*}
	\begin{split}
	&\Gamma^C_{i}\left( N,v,P\right) =ESD2^U_i(N,v,P)\\
	&+\frac{1}{2^{m-1}}\frac{1}{2^{p_{k}-1}}\left(
	\sum_{R\subseteq M\backslash k}\sum_{T\subset P_{k}, i\in T}\frac{v^{0'}(\bigcup\limits_{r\in
			R}P_{r}\cup T)}{t}-\sum_{R\subseteq M\backslash k}\sum_{T\subset P_{k}, i\notin T, T\neq \emptyset}\frac{v^{0'}(\bigcup\limits_{r\in
			R}P_{r}\cup T)}{p_k-t}\right)\\
	&+\frac{1}{2^{m-1}}\frac{1}{p_k}
	\left(\sum_{R\subset M, k\in R}\frac{v^{0'}(\bigcup\limits_{r\in
			R}P_{r})}{r}-\sum_{R\subseteq M\backslash k}\frac{v^{0'}(\bigcup\limits_{r\in
			R}P_{r})}{m-r}\right).
	\end{split}
	\end{equation*}
	
	\section*{Acknowledgements}
	
	This work has been supported by the ERDF, the MINECO/AEI grants MTM2017-87197-C3-1-P, MTM2017-87197-C3-3-P, and by the Xunta de Galicia (Grupos de Referencia Competitiva ED431C-2016-015 and ED431C-2017-38 and Centro Singular de Investigación de Galicia ED431G/01).
	
	\section*{References}
	Albino D, Scaruffi P, Moretti S, Coco S, Truini M, Di Cristofano C, Cavazzana A, Stigliani S, Bonassi S, Tonini GP (2008). Identification of low intratumoral gene expression heterogeneity in neuroblastic tumors by genome-wide expression analysis and game theory. Cancer 113, 1412-1422.\\
	Alonso-Meijide JM, Fiestras-Janeiro MG (2002). Modification of the Banzhaf value for games with a coalition structure. Annals of Operations Research 109, 213-227.\\
	Alonso-Meijide JM, Carreras F, Fiestras-Janeiro MG, Owen, G. (2007). A comparative axiomatic characterization of the Banzhaf-Owen coalitional value. Decision Support Systems 43, 701-712.\\
	Alonso-Meijide JM, Costa J, García-Jurado I (2019a). Null, Nullifying, and Necessary Agents: Parallel Characterizations of the Banzhaf and Shapley Values. Journal of Optimization Theory and Applications 180, 1027-1035.\\
	Alonso-Meijide JM, Costa J, García-Jurado I (2019b). Values, Nullifiers and Dummifiers. In: Handbook of the Shapley Value (E Algaba, V Fragnelli, J Sánchez-Soriano eds.), CRC Press, 75-92.\\
	Alonso-Meijide JM, Costa J, García-Jurado I, Gon\c{c}alves-Dosantos JC (2020). On egalitarian values for cooperative games with a priori unions. TOP. https://doi.org/10.1007/s11750-020-00553-2\\
	Amer R, Carreras F, Gimenez JM (2002). The modified Banzhaf value for games with a coalition structure: an axiomatic characterization. Mathematical Social Sciences 43, 45-54.\\
	Béal S, Navarro F (2020). Necessary versus equal players in axiomatic studies. Preprint\\
	Benati S, López-Blázquez F, Puerto J (2019). A stochastic approach to approximate values in cooperative games. European Journal of Operational Research 279, 93-106.\\
	Bergantiños G, Valencia-Toledo A, Vidal-Puga J (2018). Hart and Mas-Colell consistency in PERT problems. Discrete Applied Mathematics 243, 11-20.\\
	Carreras F, Puente A (2015). Coalitional multinomial probabilistic values. European Journal of Operational Research 245, 236-246.\\
	Casajus A (2010).  Another characterization of the Owen value without the additivity axiom. Theory and Decision 69, 523-536.\\
	Costa J (2016). A polynomial expression of the Owen value in the maintenance
	cost game. Optimization 65, 797-809.\\
	Driessen TSH, Funaki Y (1991). Coincidence of and collinearity between
	game theoretic solutions. OR Spectrum 13, 15-30.\\
	Fragnelli V, Iandolino A (2004). Cost allocation problems in urban solid wastes
	collection and disposal. Mathematical Methods of Operations Research 59, 447-463.\\
	Gonçalves-Dosantos JC, García-Jurado I, Costa J (2020). Sharing delay costs in stochastic scheduling problems with delays. 4OR-A Quarterly Journal of Operations Research. https://doi.org/10.1007/s10288-019-00427-9\\
	Gurram P, Kwon H, Davidson C (2016). Coalition game theory-based feature subspace selection for hyperspectral classification. IEEE Journal of Selected Topics in Applied Earth Observations and Remote Sensing 9, 2354-2364.\\
	Ju Y, Borm P, Ruys P (2007). The consensus value: a new solution concept for cooperative games. Social Choice and Welfare 28, 685-703.\\
	Khmelnitskaya AB, Yanovskaya EB (2007). Owen coalitional value without additivity axiom. Mathematical Methods of Operations Research 66, 255-261.\\
	Lorenzo-Freire S (2016). On new characterizations of the Owen value. Operations Research Letters 44, 491-494.\\
	Lucchetti R, Moretti S, Patrone F, Radrizzani P (2010). The Shapley and Banzhaf values in microarray games. Computers and Operations Research 37, 1406-1412.\\
	Moretti S, Patrone F (2008). Transversality of the Shapley value. Top 16, 1-41.\\
	Owen G (1975). Multilinear extensions and the Banzhaf value. Naval Research Logistics Quarterly 22, 741-750.\\
	Owen G (1977). Values of games with a priori unions. In: Mathematical Economics and Game Theory (R Henn, O Moeschlin eds.), Springer, 76-88.\\
    Owen G (1982). Modification of the Banzhaf-Coleman index for games with a priori unions. In: Power, Voting and Voting Power (MJ Holler ed.), 232-238.\\
    Ruiz LM, Valenciano F, Zarzuelo JM (1998). The family of least square values for transferable utility games. Games and Economic Behavior 24, 109-130.\\
	Schmeidler D (1969). The nucleolus of a characteristic function form game. SIAM Journal on Applied Mathematics 17, 1163-1170.\\
	Shapley LS (1953). A value for n-person games. In: Contributions to the Theory of Games II (HW Kuhn, AW Tucker eds.), Princeton University Press, 307-317.\\
	Strumbelj E, Kononenko I (2010). An efficient explanation of individual classifications using game theory. Journal of Machine Learning Research 11, 1-18.\\
	Tijs S (1981). Bounds for the core and the $\tau$-value. In: Game Theory and Mathematical Economics (O Moeschlin, D Pallaschke eds.), North Holland, 123-132.\\
	Vázquez M, van den Nouweland A, García-Jurado I (1997). Owen's coalitional value and aircraft landing fees. Mathematical Social Sciences 34,  273-286.\\
	Wang W, Sun H, van den Brink R, Xu G (2019). The family of ideal values for cooperative games. Journal of Optimization Theory and Applications 180, 1065-1086.\\
	Yang Z, Zhang Q (2015). Resource allocation based on DEA and modified Shapley value. Applied Mathematics and Computation 263, 280-286.\\

\end{document}